\documentclass{article}
\usepackage{shortcuts}
\usepackage{threeparttable}
\usepackage{algorithm}
\usepackage{hyperref}
\usepackage[capitalise]{cleveref}
\usepackage[noend]{algpseudocode}
\usepackage{makecell}

\usepackage{pgf,tikz,pgfplots}
\pgfplotsset{compat=1.14}
\usepackage{mathrsfs}
\usetikzlibrary{arrows}

\algnewcommand\algorithmicforeach{\textbf{for each}}
\algdef{S}[FOR]{ForEach}[1]{\algorithmicforeach\ #1\ \algorithmicdo}

\newtheorem{lemma}{Lemma}
\newtheorem{theorem}[lemma]{Theorem}
\newtheorem{corollary}[lemma]{Corollary}
\newtheorem{definition}[lemma]{Definition}

\newtheorem{informal theorem}[lemma]{Informal Theorem}
\newtheorem{remark}[lemma]{Remark}

\usepackage{tikz}
\usetikzlibrary{arrows}
\date{}
\title{Graph Spanners in the Message-Passing Model}
\author{
    Manuel Fernández V\\Carnegie Mellon University\\ \texttt{manuelf@andrew.cmu.edu} \and
    David P. Woodruff\\Carnegie Mellon University\\ \texttt{dwoodruf@cs.cmu.edu} \and
    Taisuke Yasuda\\Akuna Capital\\ \texttt{yasuda.taisuke1@gmail.com}
}

\begin{document}

\begin{titlepage}
\maketitle
\thispagestyle{empty}
\begin{abstract}
Graph spanners are sparse subgraphs which approximately preserve all pairwise shortest-path distances in an input graph. The notion of approximation can be additive, multiplicative, or both, and many variants of this problem have been extensively studied. We study the problem of computing a graph spanner when the edges of the input graph are distributed across two or more sites in an arbitrary, possibly worst-case partition, and the goal is for the sites to minimize the communication used to output a spanner. We assume the message-passing model of communication, for which there is a point-to-point link between all pairs of sites as well as a coordinator who is responsible for producing the output. We stress that the subset of edges that each site has is not related to the network topology, which is fixed to be point-to-point. While this model has been extensively studied for related problems such as graph connectivity, it has not been systematically studied for graph spanners. We present the first tradeoffs for total communication versus the quality of the spanners computed, for two or more sites, as well as for additive and multiplicative notions of distortion. We show separations in the communication complexity when edges are allowed to occur on multiple sites, versus when each edge occurs on at most one site. We obtain nearly tight bounds (up to polylog factors) for the communication of additive $2$-spanners in both the with and without duplication models, multiplicative $(2k-1)$-spanners in the with duplication model, and multiplicative $3$ and $5$-spanners in the without duplication model. Our lower bound for multiplicative $3$-spanners employs biregular bipartite graphs rather than the usual Erd\H{o}s girth conjecture graphs and may be of wider interest.
\end{abstract}
\end{titlepage}

\section{Introduction}
%
In modern computational settings, graphs are often stored in a distributed setting with edges living across multiple servers. This may happen when traditional, single-server methods for representing and processing massive graphs are no longer feasible and require parallel processing capability to complete. In other real world settings, different sites collect information in different locations, naturally leading to a computational setting with an input graph distributed across servers. For example,
the sites may correspond to sensor networks, different network servers,
etc.
%
Furthermore, the bottleneck in these settings is often in the \emph{communication} between the servers, rather than the computation time within each of the servers. Computing synopses of distributed graphs in a communication-efficient
manner has therefore become increasingly important.

We consider the problem of efficiently
constructing a {\it graph spanner} in the
message-passing model of communication. A graph spanner is a subgraph of
the input graph, for which shortest path distances
are approximately preserved in the subgraph. This property can immediately be used to approximately answer shortest path queries, diameter queries, connectivity queries, etc. 
Spanners have applications to internet routing  \cite{TZ01,c01,cw04,pu88},
using protocols in unsynchronized networks to simulate synchronized networks \cite{pu89a},
distributed and parallel algorithms for shortest paths \cite{c98,c00,e05},
and for constructing distance oracles \cite{tz05,bs04}.
There are various notions of approximation provided by a spanner,
such as {\it additive}, for which there is an integer $\beta \geq 1$ and one wants for all pairs $u,v$ of vertices, that $d_H(u,v) \leq d_G(u,v) + \beta$, as well
as {\it multiplicative}, in which case there is an integer $\alpha \geq 1$ and one wants for all pairs $u,v$ of vertices, that $d_H(u,v) \leq \alpha \cdot d_G(u,v)$.

\paragraph{Message-Passing Model.}

In the message-passing model (see, e.g., \cite{phillips2012lower,wz12,braverman2013tight,woodruffzhang13,HuangRVZ15}) there are $s$ players, denoted $P^{1}, P^{2}, \dots, P^{s}$, and each player holds part of the input. In our context, player $P^i$ holds a subset $E_i$ of a set of edges on a common vertex set $V$, and we define the graph $G$ with vertex set $V$ and edgeset $\bigcup_i E_i$. There are various input models, such as the {\it without duplication} edge model in which the $E_i$ are pairwise disjoint, and the {\it with duplication} edge model in which the $E_i$ are allowed to overlap.
In this model there is also a coordinator $C$ who is required to compute a function defined on the union of the inputs of the players.
The communication channels in this model are point-to-point. For example, if $C$ is communicating with $P^{i}$, then the remaining $s-1$ players do not see the contents of the message between $C$ and $P^{i}$.
We also do not allow the players to talk directly with each other; rather, all communication happens between the coordinator and a given player at any given time\footnote{This model only mildly increases the communication cost over a complete
point-to-point network in which each pair of players can communicate with
each other. Indeed, if $P^i$ wishes to speak to $P^j$, then $P^i$ can forward
a message through the coordinator who can send it to $P^j$.
Thus the communication increases by a multiplicative factor of $2$
and an additive $O(\log s)$
per message to specify where to forward the message. As these
factors are small in our context, we will focus on the coordinator model.}. The coordinator $C$ is responsible for producing the output.

The main resource measure we study
is the {\it communication complexity},
that is, the total number of bits required to be sent between the servers in order to output such a spanner with high probability.
While graph spanners have been studied in the offline model, as well as in various distributed models such as the CONGEST and LOCAL models, e.g., \cite{censor2016distributed,en17,dn17,gp17,Censor-HillelD18,ParterY18} as well as in the local computation algorithms model \cite{ParterRVY19}, they have not been systematically studied in the message-passing model.
The few related results we are aware of in the message-passing model
are given in \cite{woodruffzhang13}, where (1)
the problem of testing graph connectivity was studied, which can be viewed as a very
special case of a spanner, and (2) a result on additive $2$-spanners which we
discuss more and improve upon below.
There is also work in related models such as \cite{knpr15},
but such models require that
the edges be randomly distributed, which may not be a realistic assumption in certain
applications, e.g., if data is collected at sensors with different
input distributions.

We also study a variant of the communication complexity in the message-passing model known as the \emph{simultaneous communication complexity} for the multiplicative $(2k-1)$-spanner problem, in which each server is only allowed to send one round of communication to the coordinator \cite{BabaiGKL03, WeinsteinW15}.

\paragraph{Turnstile Streaming Model.}

Finally, we record some simple results in the \emph{turnstile streaming model}, in which the input graph is presented as a stream of insertion and deletion updates of edges. That is, we view our graph as an $\binom{n}{2}$-dimensional vector $x$ starting with the zero vector, and we receive updates of the form $(e_i, \Delta_i)\in[\binom{n}2]\times \{\pm 1\}$ and increment the $e_i$th entry of $x$ by $\Delta_i$. Our input graph is then the graph that has the edge $e$ iff $\sum_{i : e_i = e}\Delta_i > 0$. We assume that the input graph has no self-loops. In this model, we wish to design algorithms using low space and low number of passes through the stream. The study of graph problems in this model were pioneered by \cite{AhnGM12a} and were subsequently studied by many other works, including \cite{AhnGM12b, AhnGM13, kapralov2014spanners, KapralovLMMS17, kapralov2019dynamic}.

\subsection{Our Results}
\begin{table}
\centering
\begin{threeparttable}
\begin{tabular}{|c|c|c|c|c|} \hline
    &\multicolumn{2}{|c|}{ With duplication }& \multicolumn{2}{|c|}{Without duplication} \\ \hline
    Problem & LB & UB & LB &UB\\ \hline
    $+2$-spanner & $\Omega(s n^{3/2})$ & $\tilde{O}(s n^{3/2})$ & $\Omega(\sqrt{s}n^{3/2} + sn)$ & $\tilde{O}(\sqrt{s} n ^{3/2}+sn)$ \\ \hline
    $+4$-spanner & $\Omega(s n^{4/3})$ & $\tilde{O}(s n^{3/2})$ & $\Omega(s^{1/3}n^{4/3} + sn)$ & $\tilde{O}(\sqrt{s} n ^{3/2}+sn)$ \\ \hline
    $+k$-spanner & $\Omega(sn^{4/3})$ & $\tilde{O}(sn^{3/2})$ & $ \Omega(n^{4/3}  + sn)$ & $\tilde{O}(\sqrt{s/k} n ^{3/2}+snk)$ \\ \hline
    $\times 3$-spanner & $\Omega(sn^{3/2})$ & $\tilde O(sn^{3/2})$ & $\Omega(s^{1/2}n ^ {3/2} + s n)$ & $\tilde O(s^{1/2}n^{3/2} + sn)$ \\ \hline
    $\times 5$-spanner & $\Omega(sn^{4/3})$ & $\tilde O(sn^{4/3})$ & $\Omega(s^{1/3}n ^ {4/3} + s n)$ & $\tilde O(s^{1/3}n^{4/3} + sn)$ \\ \hline
    \makecell{$\times (2k-1)$-spanner, \\ $k\geq 3$}
    & $\Omega(sn^{1 + 1/k})$ & $\tilde O(sn^{1 + 1/ k})$ & $\Omega(s^{1/2-1/2k}n ^ {1 + 1/k} + s n)$ & $\tilde O(ks^{1-2/k}n^{1 + 1/ k}+snk)$ \\ \hline
    \makecell{$\times (2k-1)$-spanner, \\ (simultaneous)}
    & $\Omega(sn^{1+1/k})$ & $\tilde O(sn^{1+1/k})$ & $\Omega(sn^{1+1/k})$ & $\tilde O(sn^{1+1/k})$ \\ \hline
\end{tabular}
\caption{Our results.}\label{table:main}
\end{threeparttable}
\end{table}
We summarize our results in Table \ref{table:main}. Note that the
$\tilde{O}$ and $\tilde{\Omega}$ notation hides $\text{poly}(\log n)$ factors. Often our
upper bounds are stated in terms of edges, but since each edge
can be represented using $O(\log n)$ bits, we obtain the same
upper bound in terms of bits up to an $O(\log n)$ factor.
We study both the with duplication
and without duplication edge models, and in all cases we consider a worst-case
distribution of edges.

We give a number of communication
versus approximation quality tradeoffs for additive spanners and multiplicative spanners.
We describe each type of spanner we consider in the
sections below, together with the results that we obtain. We obtain qualitatively
different results
depending on whether edges are allowed to be duplicated across the players, or if
each edge is an input to exactly one player.

We point out some particular notable aspects of our results. First, we obtain nearly tight bounds (up to $\poly(\log n)$ factors) for the communication of additive $2$-spanners in both the with and without duplication models, multiplicative $(2k-1)$-spanners in the with duplication model, and multiplicative $3$ and $5$-spanners in the without duplication model. Second, in proving our tight lower bound for $3$-spanners in the without duplication model (\cref{thm:mult-3-lb}), we employ results from extremal graph theory on \emph{biregular bipartite graphs}, which, to the best of our knowledge, is the first explicit use of such graphs in the context of lower bounds for spanners. All other lower bounds that we are aware of are obtained from extremal graphs given by the Erd\H{o}s girth conjecture (e.g., lower bounds in the streaming \cite{Baswana08}, local computation algorithm \cite{ParterRVY19}, and distributed \cite{censor2016distributed,en17,dn17,gp17,Censor-HillelD18,ParterY18} models), and we believe that our use of biregular bipartite graphs may inspire tight lower bounds in other models in the future as well.

We note that our results slightly differ from traditional results on spanners, in that the sparsity of our spanner may be far from optimal. For instance, we show an algorithm for computing an additive $2$-spanner in the without duplication model with near-optimal communication complexity of $\tilde O(\sqrt{s}n^{3/2})$ bits of communication, but the size of this spanner is $\tilde O(\sqrt{s}n^{3/2})$ edges, which may be much larger than the optimal $O(n^{3/2})$ edges when the number of servers $s$ is very large. It is an interesting question to characterize the communication complexity of computing spanners of optimal size.

\subsubsection{Additive Spanners}
In the case of additive spanners, one is given an arbitrary graph $G$
on a set $V$ of $n$ vertices and an integer parameter $\beta \geq 1$,
and we want to output a subgraph $H$ containing as few edges as possible
so that $d_H(u,v) \leq d_G(u,v) + \beta$ for all pairs of vertices
$u,v \in V$. The first such spanner was constructed by Aingworth et al. \cite{aingworth1999fast},
which was slightly improved in \cite{dor2000all,ep04}. They showed, surprisingly,
that for $\beta = 2$, it is always possible to achieve $|H| = O(n^{3/2})$. The next
additive spanner was constructed in \cite{BKMP05}, where it was shown that for $\beta = 6$
one can achieve $O(n^{4/3})$ edges; see also \cite{woodruff2010additive} where the time
complexity was optimized. Recently, it was shown in \cite{c13} how to achieve an
additive spanner with $\tilde{O}(n^{7/5})$ edges for $\beta = 4$.
In a breakthrough work \cite{abboud20164},
an $\Omega(n^{4/3-o(1)})$ lower bound was shown for any constant $\beta$.

The one previous result we are aware of for computing spanners in the message-passing model is
for additive $2$-spanners given in \cite{woodruffzhang13}, for which an $\tilde{O}(sn^{3/2})$ upper bound
was given which works with edge duplication. We first show that with edge duplication,
the algorithm of \cite{woodruffzhang13} is {\it optimal}, by proving a matching $\Omega(sn^{3/2})$ lower bound. Our
lower bound is a reduction from the $s$-player set disjointness problem \cite{braverman2013tight}.
We next consider the
case when there is no edge duplication, and perhaps surprisingly, show that one can achieve an additive $2$-spanner with $\tilde{O}(\sqrt{s} n^{3/2})$ communication, improving upon the $\tilde{O}(sn^{3/2})$
bound of \cite{woodruffzhang13}, and given our lower bound in the case of edge duplication, providing a separation
for additive spanners in the models with and without edge duplication. Our upper bound is based on observing
that the dominant cost in implementing additive spanner algorithms in a distributed setting is that of
performing a breadth-first search. We instead perform fewer breadth first searches to obtain a better overall
communication cost than one would obtain by na\"ively implementing an offline additive spanner algorithm,
as is done in \cite{woodruffzhang13}. This algorithm is the starting point for our technically
more involved upper bound, where we show
that it is possible to obtain an additive $k$-spanner with $\tilde{O}(\sqrt{s/k}n^{3/2} + snk)$ total
communication. We complement this result with a lower bound of $\Omega(sn^{4/3-o(1)})$ for this problem.

We note that we are not able to obtain constant additive spanners with fewer than $n^{3/2}$ edges,
as the dominant cost comes from having to do breadth first search trees, which is communication-intensive
in the message-passing model. We conjecture that
$\Theta(n^{3/2})$
may be the optimal communication bound for any additive spanner with constant distortion, unlike in the offline
model where an $O(n^{4/3})$ edge bound is achievable.

\subsubsection{Multiplicative Spanners}
In the case of multiplicative spanners,
we are given an arbitrary graph $G$ on a set $V$ of
$n$ vertices and an integer parameter
$\alpha \geq 1$, and wish to output a subgraph $H$ containing as few edges
as possible so that $d_H(u,v) \leq \alpha \cdot d_G(u,v)$ for all pairs of vertices
$u,v \in V$. For odd integers $\alpha = 2k-1$, for any graph $G$ on
$n$ vertices there exists
a $\alpha$-spanner with $O(n^{1+1/k})$ edges, for any integer $k \geq 1$ \cite{a85}.
Further,
this is known to be optimal for $k \in \{1, 2, 3, 5\}$ \cite{Tits59, w91}, while
for general $k$ the best known bounds are $\Omega(n^{1+2/(3k-3)})$ for
odd $k$ and $\Omega(n^{1+2/(3k-2)})$ for even $k$ \cite{l95,l96}.

Under a standard conjecture of Erd\H{o}s \cite{e65}, this bound of $O(n^{1+1/k})$ is in fact
optimal for every $k$. Recall that the girth of an unweighted graph
is the minimum length cycle in the graph. Erd\H{o}s's conjecture is
that there exist graphs $G$
with $\Omega(n^{1+1/k})$ edges for which the girth is $2k+2$. Note that
given such a $G$, if one were to delete any edge $\{u,v\}$ in $G$,
then the distance from $u$ to $v$ would increase from $1$ to
$2k+1$, and therefore $G$ is the only $2k-1$-spanner of itself,
giving the $\Omega(n^{1+1/k})$ edge lower bound. Notice that $G$ is also
the only $2k$-spanner of itself, and so the $\Omega(n^{1+1/k})$ lower bound
also holds for even integers $\alpha = 2k$, which is also optimal since,
as mentioned above, there always exist $(2k-1)$-spanners with $O(n^{1+1/k})$ edges.

\paragraph{Message-Passing Model.}

We show that for computing a multiplicative $(2k-1)$-spanner
with $s$ players, in the edge model with duplication on $n$-node graphs,
there is an $\Omega(s \cdot OPT_k)$ communication lower bound, where $OPT_k$ is the maximum size
of a $(2k-1)$-spanner of any graph. Our lower bound is again based on a reduction from the multiplayer
set disjointness communication problem. 
A greedy algorithm shows that this bound is optimal, that is, we provide a matching $\tilde O(s \cdot OPT_k)$ upper bound.

If instead each edge occurs on exactly one server, note that the additive $2$-spanner algorithm already gives a separation in the $s$ parameter by providing a $\tilde O(\sqrt{s}n^{3/2}+sn)$ algorithm. We show that this is optimal up to polylog factors by showing a lower bound of $\Omega(\sqrt{s}n^{3/2})$ for multiplicative $3$-spanners. This then gives near optimal lower bounds for additive $2$-spanners as well. Our lower bound here uses for the first time, to the best of our knowledge, the theory of \emph{biregular bipartite cages}, which may be of wider interest. For $k\geq 3$, we again show that there is a separation in the $s$ parameter between the models with and without edge duplication, by showing that carefully balancing the complexity of a lesser known variant of the classic algorithm of \cite{baswana2007simple}, the cluster-cluster joining variant, can be implemented to use only $\tilde O(ks^{1-2/k}n^{1+1/k}+snk)$ communication. We complement this result with a lower bound of $\Omega(s^{1/2-1/2k}n^{1+1/k}+sn)$ communication via a reduction from the edge model with duplication, essentially by splitting vertices to transform the input instance with duplication into one without duplication. This bound is off by a factor of $O(s^{1/2-3/2k})$. For $k=3$, the exponent on $s$ is exactly correct, giving a nearly tight characterization of $\tilde\Theta(s^{1/3}n^{4/3})$ communication for the problem of computing multiplicative $5$-spanners.

\paragraph{Simultaneous Communication.}

In the simultaneous communication model, we show an upper bound of $\tilde O(sn^{1+1/k})$ in the with duplication model and a lower bound of $\Omega(sn^{1+1/k})$ without duplication model under the Erd\H{o}s girth conjecture, showing that the complexity is $\tilde\Theta(sn^{1+1/k})$ in all cases. The upper bound simply comes from locally computing a multiplicative $(2k-1)$-spanner of size $\Theta(n^{1+1/k})$ at each server, while the lower bound comes from constructing $s$ edge-disjoint graphs on $n$ vertices and $\Omega(n^{1+1/k})$ edges, a constant fraction of which must be sent to the server in the simultaneous communication model, as we show.

\paragraph{Turnstile Streaming Model.}

Finally, we note that implementing the cluster-cluster joining algorithm of \cite{baswana2007simple} in the turnstile streaming model gives an algorithm for computing a multiplicative $(2k-1)$-spanner with $(\floor{k/2}+1)$ passes and $\tilde O(n^{1+1/k})$ space. Our algorithm follows the techniques of \cite{AhnGM12b}, but implements a different version of the Baswana-Sen algorithm than they do, which allows us to save on the number of passes. Previously, in the regime of a small constant number of passes, \cite{kapralov2014spanners} gave an algorithm for computing multiplicative spanners with distortion $2^k$ in $\tilde O(n^{1+1/k})$ space with two passes. Our result improves upon this in the distortion for $k = 3$, achieving an optimal space-distortion tradeoff.

 \section{Preliminaries}
We use $[n]$ to denote $\{1, \dots, n\}$. We often use capital letters  $X$, $Y$, $\dots$ for sets, vectors, or random variables, and lower case letters $x$, $y$, $\dots$ for specific values of the random variables $X$, $Y$, $\dots$.
For a set $S$, we use $|S|$ to denote the size of $S$.

As for messages and communication, we assume that all communication is measured in terms of bits.
All logarithms in this paper are base 2.

We make use of the Set Disjointness problem in the message-passing model,
see, e.g., \cite{chattopadhyay2010story}.

\begin{definition}[$\mathrm{DISJ}_{n,s}$]
There are $s$ players and each of them holds a set $X_{i} \subseteq [n]$, and the goal is to determine whether $\bigcap_{i = 1} ^{s} X_{i}$ is empty or not.
\end{definition}
Recently in \cite{braverman2013tight}, the authors obtained a tight lower bound for this problem.
\begin{theorem}[\cite{braverman2013tight}, Theorem 1.1]
	\label{thm:setdisjointness}
	For every $\delta > 0$, $n \geq 1$ and $ s = \Omega (\log n)$, the randomized communication complexity of
set disjointness in the message-passing model is $\Omega(sn)$ bits.
That is, for every randomized protocol which succeeds with probability
at least $2/3$ on any given set of inputs, there exists a set of inputs
and random coin tosses of the players which causes the sum of message
lengths of the protocol to be $\Omega(sn)$ bits. Further, for any $s \geq 2$, the randomized
communication complexity of set disjointness is $\Omega(n)$ (\cite{ks92}).
\end{theorem}

Let $G=(V,E)$ be an undirected graph, where $V$ is the vertex set and $E$ is the edgeset.
Let $n = |V|$ and $m = |E|$ denote the number of vertices and the number of edges, respectively.
For a pair of vertices $u,v$ in $G$, the distance between $u$ and $v$ is denoted by $d_{G}(u,v)$, which
indicates the length of the shortest path connecting $u$ to $v$.
The results in this paper are for unweighted graphs, thus the length of a path is equal to the number of edges is contains.

In the message-passing model, we have $s$ players.
We suppose each player knows the entire vertex set $V$ and a subset $E_i$ of the input
graph,
where $E_{i}$ is a subset of the edge set $E$.
We can think of each player $P_{i}$ having a bit vector $Y_{i}$, which is a vector of length $m$.
We number the edges with $1$, $2, \ldots, m$.
If $Y_{i,j} = 0$, it indicates that the $j$-th edge is missing in $P_{i}$.
If $Y_{i,j} = 1$, it indicates that the $j$-th edge is present in $P_{i}$.

We study two models: allowing edge duplication and not allowing it.
In the model {\it with duplication}, edgesets can overlap.
For most problems with duplication, we will obtain lower bounds on their communication
via a reduction from the Set Disjointness problem.
In the model {\it without duplication}, all edge sets are disjoint,
that is, $\forall i, j \in [s], i \neq j, E_{i} \cap E_{j} = \varnothing$.

We will also assume in this paper that $s \ll n$, e.g.\ $s = O(n^\eps)$ for a small constant $\eps$: this is typically the case in practice, as well as the interesting regime for most of our bounds.

\section{Additive Spanners}
In this section we study how to compute additive spanners of graphs in the message-passing model. Recall the definition of additive spanners.
\begin{definition}[Additive spanners]
Given a graph $G$, a subgraph $H$ is an \emph{additive $\beta$-spanner for $G$} if for all $u, v \in V$, $d_{G}(u,v) \leq d_{H}(u,v) \leq d_{G}(u,v) + \beta,$ where $d_{G}(u,v)$ and $d_{H}(u,v)$ are lengths of the shortest paths in $G$ and $H$, respectively.
\end{definition}


\subsection{Additive \texorpdfstring{$2$}{2}-Spanners with Duplication}

As a warmup, we first consider the case when $\beta = 2$, and edge duplication is allowed. A large fraction of our proofs will follow this paradigm.

\begin{theorem} \label{thm:42}
	The optimal communication cost of the additive $2$-spanner problem with edge duplication in the message passing model is $\tilde{\Theta}(s n^{3/2})$ bits.
\end{theorem}

The following lemma is well known.

\begin{lemma} \label{lemma:6cyclefree}
For every $n$, there is a family of graphs on $n$ vertices with $\Theta(n ^{3/2})$ edges and girth at least $6$.
\end{lemma}
\begin{proof}
	By a special case of Theorem 2 from \cite{furedi1994quadrilateral}, there is a family of graphs with $\frac{\sqrt{2}}{2} n^{3/2} + \Omega(n^{4/3}) = \Theta(n^{3/2})$ edges containing no $4$-cycles.

	Color vertices red and blue uniformly at random, and keep only edges between red and blue vertices. The resulting graph is bipartite and does not contain cycles of length $4$, therefore it does not have cycles of length $5$ either. Thus it has girth at least $6$. Since we keep half of the edges in expectation, there exists a graph with $\Theta(n^{3/2})$ edges with girth at least $6$.
\end{proof}

We also show the following very general lemma which we will make use of several times:

\begin{lemma}\label{lemma:propertysubgraph}
	Let $R$ be a binary relation between graphs and members of a set $\mathcal{P}$. Suppose there is a family of graphs $\{G_n\}_{n}$ such that $G_n$ has $n$ vertices and $f(n)$ edges, and:
	\begin{enumerate}
		\item $p_n$ is the unique member of $\mathcal{P}$ with $(G_n, p_n) \in R$
		\item for any proper subgraph $H$ of $G_n$, $(H, p_n) \not \in R$
	\end{enumerate}
	Then for a graph $G$ on $n$ vertices, the communication complexity in the edge duplication case of computing $p$ such that $(G, p) \in R$ is $\Omega(sf(n))$ bits.
\end{lemma}

For concreteness, in this example we may think of $\mathcal{P}$ as the set of all graphs, and define $R$ to be the set of pairs $(G,S)$ such that $S$ is an additive 2-spanner of $G$.

\begin{proof}
	We reduce from the set disjointness problem in the message-passing model. Given an instance of set disjointness with $s$ players each holding $X_i \subseteq[f(n)]$, we create a graph $G_n$ on $n$ vertices. We give player $i$ the edge indexed by $j$ if $j \not \in X_i$. If the coordinator outputs $p = p_n$, we output that $\bigcap_i X_i \neq \varnothing$, otherwise we output that $\bigcap_i X_i = \varnothing$. The coordinator outputs $p_n$ if and only if all the edges of $G_n$ are present among the players, which is the case if and only if $\bigcap_i X_i \neq \varnothing$. Therefore this procedure correctly decides set disjointness. Theorem \ref{thm:setdisjointness} implies a $\tilde{\Omega}(sf(n))$ bit lower bound for the communication cost of computing $p$.
\end{proof}

Together, these pieces yield the following:
\begin{proof}[Proof of Theorem \ref{thm:42}]
	For the lower bound, we observe that for a graph $G$ as in Lemma \ref{lemma:6cyclefree} removing any edge $(u,v)$ increases the distance from $u$ to $v$ to at least 5, and thus the only additive-$2$ spanner of $G$ is $G$ itself. By Lemma \ref{lemma:propertysubgraph} with $\mathcal{P}$ as the set of all graphs and $R$ as the set of pairs $(G,H)$ such that $H$ is an additive 2-spanner of $G$, we immediately have that the communication cost of finding a subgraph that is an additive $2$-spanner is $\tilde{\Omega}(s|E(G)|) = \tilde{\Omega}(sn^{3/2})$.

	For the upper bound, one can show that the well known algorithm of \cite{dor2000all} for computing additive $2$-spanners can be implemented in the message passing model with $\tilde{O}(s n^{3/2})$ bits of communication, even in the case of edge duplication. See the proof of Theorem 5 of \cite{woodruffzhang13} for details.
\end{proof}

\subsection{Additive \texorpdfstring{$k$}{k}-Spanners with Duplication}
Unfortunately, we are not able to design algorithms with improved communication over the above additive $2$-spanners even if we allow for larger additive distortion, despite the existence of algorithms for additive $6$-spanners that achieve $O(n^{4/3})$ edges \cite{BKMP05, woodruff2010additive}. In this section, we show a lower bound of $\Omega(sn^{4/3-o(1)})$ on the communication of additive $k$-spanners via a similar argument to the lower bound in \cref{thm:42}.

\begin{theorem}\label{thm:additive-lb}
The randomized communication complexity of the additive $k$-spanner problem with edge duplication is $\Omega(sn^{4/3-o(1)})$.
\end{theorem}
\begin{proof}
The proof follows essentially from applying \cref{lemma:propertysubgraph} on the extremal graph of \cite{abboud20164}, with minor modifications. The details are deferred to \cref{section:additiveNoDupLB}.
\end{proof}

\subsection{Additive \texorpdfstring{$2$}{2}-Spanners without Duplication}
We next show how to improve the upper bound of Theorem \ref{thm:42} when edges are not duplicated across servers. We note that we can assume all servers know the degree of every vertex, since this involves exchanging at most $n$ numbers per player or $O(ns \log n)$ bits of communication. This is negligible compared to the rest of the communication assuming $s \ll n$.

First we write down some simple lemmas that we will make use of multiple times. The proofs of these can be found in \cref{section:simple-lemmas}.
\begin{lemma}\label{lemma:samplecover}
	Let $\mathcal{C}$ be a collection of sets over a ground set $\mathcal{U}$ each of size at least $t$. If we sample $\frac{|\mathcal{U}|}{t} \log |\mathcal{C} / \delta|$ elements from $\mathcal{U}$ uniformly with replacement, with probability at least $1-\delta$ we sample at least one element from each set in $\mathcal{C}$.
\end{lemma}

\begin{lemma}\label{lemma:bfs_communication}
	The deterministic communication complexity of computing a BFS (breadth first search) tree from a given node in the message passing model (with our without duplication) is $\tilde O(sn)$.
\end{lemma}

We are ready to state the main algorithm of this section:
\begin{theorem} \label{plus2without}
The randomized communication complexity of the additive-$2$ spanner problem without edge duplication is $\tilde{O}(\sqrt{s} n ^{3/2})$.
\end{theorem}
\begin{algorithm}
\caption{$+2$ spanner without edge duplication}
\begin{algorithmic}[1]
\Require $G = (V, E)$.
\Ensure $H$, $+2$ spanner of $G$.
\State $V_{1} = \{x \in V: \text{degree of } x \leq \sqrt{sn} \}$.
\State Each player sends the coordinator all edges adjacent to $V_1$. The coordinator aggregates these and compiles the set $E_{1} = \{(u,v) \in E: u \in V , v \in V_{1} \}$. \label{line:compileE1}
\State The coordinator samples $2 \log n \cdot \sqrt{\frac{n}{s}}$ vertices uniformly at random with replacement from $V$, and let $\mathcal{R}$ denote the sampled vertex set.
\State Grow a BFS tree $T_{x}$ from each $x \in \mathcal{R}$, let $E(T_x)$ be its edge set. \label{line:growBFS}
\State $F = E_1 \cup \bigcup_{x \in \mathcal{R}} E(T_{x})$.
\State \Return $H = (V,F)$.
\end{algorithmic}
\end{algorithm}
\begin{proof}
First we will show this algorithm provides a $+2$ spanner of $G$ with constant probability.

Consider the set $V \setminus V_1$ of vertices with degree $\geq \sqrt{sn}$. Let $\mathcal{E}$ denote the event that $\mathcal{R}$ contains at least one vertex from the neighborhood of every vertex in this set. Applying Lemma \ref{lemma:samplecover} with $\mathcal{U} = V$, with $\mathcal{C}$ as the collection of neighborhoods of vertices in $V \setminus V_1$, and with $t = \sqrt{sn}$, we have that $\mathcal{E}$ occurs with probability at least $1 - o(1)$.

Now for an arbitrary pair of vertices $u,v$, let us consider the shortest path $P$ connecting them in $G$. Suppose an edge $(x,y) \in P$ is missing from $E_1$. This implies that both $x$ and $y$ are in $V \setminus V_1$ and have degree strictly larger than $\sqrt{sn}$. If $\mathcal{E}$ holds, then $x$ has a neighbor $w$ sampled in $\mathcal{R}$. Then:
\begin{eqn}
d_H(u,v) &\leq d_H(u,w) + d_H(w,v) \leq d_G(u,w) + d_G(w,v) \leq d_G(u,x) + 1 + d_G(x,v) + 1 = d_G(u,v) + 2
\end{eqn}
Above the first and third line follow from the triangle inequality, the second holds since $H$ includes a BFS-tree rooted at $w$, and the last line since $x$ is on the shortest path between $u$ and $v$.

Next we bound the communication. Line \ref{line:compileE1} requires $\tilde O(\sqrt{s}n^{3/2})$ communication since each edge is in $N(V_1)$ is sent exactly once. By Lemma \ref{lemma:bfs_communication}, growing $\log n \sqrt{n/s}$ BFS trees on line \ref{line:growBFS} requires $\tilde O(sn \cdot \sqrt{n/s} = \sqrt{s}n^{3/2})$ communication as well.
Thus the total communication is $\tilde{O}(\sqrt{s} n ^{3/2})$.
\end{proof}

We will later show that this is nearly optimal by showing a lower bound of $\Omega(\sqrt{s}n^{3/2})$ for the weaker problem of computing a multiplicative $3$-spanner in \cref{thm:mult-3-lb} later in the paper.

\subsection{Additive \texorpdfstring{$k$}{k}-Spanners without Duplication}

We now study the additive spanner problem with larger distortion in the without duplication model.

Although our additive $2$-spanner lower bound came from our multiplicative $3$-spanner lower bounds, for larger distortions, this technique does not give optimal dependence on $n$, which is the dominant variable in the parameter regime of interest. Instead, we show the following:

\begin{theorem}
	\label{thm:additiveNoDupLB}
	The randomized communication complexity of the additive $k$-spanner problem without edge duplication is $\Omega(n^{4/3 - o(1)} + sn)$.
\end{theorem}

\begin{proof}
	The lower bound essentially follows from the strong incompressibility result of additive spanners in Theorem 2 of \cite{abboud20164}. We defer the details to \cref{section:additiveNoDupLB}.
\end{proof}

It is worth noting why the communication lower bounds from set disjointness no longer hold in the setting where edges are not duplicated across players, at least if we follow the same reduction as above. Imposing the assumption of edge disjointness amounts to imposing the restriction that the set disjointness instances have the property that the complements of the bit vectors held by each player do not intersect. However, this restricted problem no longer has an $\Omega(sn)$ lower bound, since it can be decided with $O(s\log n)$ bits of communication: each player sends the size of their complement set to the coordinator, and the coordinator checks if their sum is exactly $n$.

We now turn to algorithms, showing that the communication drops off by a factor of $\sqrt{k}$ for larger additive distortions $k$.
\begin{theorem}
	The randomized communication complexity of the additive $k$-spanner problem without edge duplication is $\tilde{O}(\sqrt{s/k}n^{3/2} + snk)$.
\end{theorem}
\begin{algorithm}
\caption{$+k$ spanner without edge duplication}
\begin{algorithmic}[1] \label{swithout}
\Require $G = (V, E)$
\Ensure $H$, $+s$ spanner of $G$
\State $V_{1} = \{x \in V: \text{degree of } x \leq \sqrt{sn/k} \}$, $E_{1} = \{(u,v) \in E: u \in V , v \in V_{1} \}$
\State Uniformly sample $\tilde{O} (\sqrt{n/sk}) + \tilde{O}(k)$ vertices from $V$, and let $\mathcal{R}_1$ denote the set of sampled vertices. \label{line:sampleline}
\State Grow a BFS tree $T_{x}$ from every $x \in \mathcal{R}_1$, $E_2 = \{e \in E : e \in T_x \text{ for some } x \in \mathcal{R}_1 \}$. \label{line:fullBFS}
\State Uniformly sample $\tilde{O}(\sqrt{kn/s})$ vertices from $V$, and let $\mathcal{R}_2$ denote the sampled vertices.
\State Grow a truncated BFS tree $T_{x}$ from every $x \in \mathcal{R}_2$, such that $|T_x| = n / k$. (In the last level of building the tree, arbitrarily include edges until $|T_x| = n/k$.) Let ${E_3 = \{e \in E : e \in T_x \text{ for some } x \in \mathcal{R}_2\}}$ \label{line:truncatedBFS}
\State $F \leftarrow E_1 \cup E_2 \cup E_3$
\State \Return $H = (V,F)$
\end{algorithmic}
\end{algorithm}

\begin{proof}
We will first show that \cref{swithout} gives an additive $k$-spanner with probability at least $1-o(1)$, then argue that it achieves the stated communication complexity. We may assume that $k \geq 6$, since otherwise \cref{plus2without} directly implies the claim. For convenience, let $N_\ell(e)$ denote the set of vertices within $\ell$ hops of either of the endpoints of $e$. Suppose we have added only the edges $E_1$ which are adjacent to vertices of degree at most $\sqrt{sn/k}$, which is the case at the end of line \ref{line:fullBFS}, and consider the shortest path $P$ in $G$ between an arbitrary pair of vertices $u,v$. Let $D$ be the set of edges of $P$ missing from $E_1$.

\begin{itemize}
	\item \textbf{Case 1}: $|D| \geq k$

	Since $P$ is a simple path and we have already included all edges adjacent to low-degree vertices, there are collectively at least $k \sqrt{sn/k}/2 = \Omega(\sqrt{snk})$ vertices in the union of the $N_1(e)$ for all $e$ missing from $P$. Let $\mathcal{E}_1$ be the event that $\mathcal R_1$ contains a vertex from this neighborhood for every choice of $u, v$ with at least $k$ missing path edges. If $\mathcal{E}_1$ holds, since in line \ref{line:truncatedBFS} we include a BFS tree from each sampled vertex, this implies that the returned $H$ is a $+2$ spanner for each such pair of $u,v$ by the same reasoning as in the proof of \cref{plus2without}.

	\item \textbf{Case 2}: $|D| < k$

	In what follows, we will argue that our construction either bridges each missing $e=(u',v')\in D$ with a $2$-hop path, or places the root of a full BFS within distance $3$ of $P$. If all $e \in D$ are bridged by 2-hop paths, we will argue that these paths are contained in truncated BFS trees included in line \ref{line:truncatedBFS}. Since there are at most $k$ edges missing from $P$, and since the distance between the endpoints of $e$ changes from $d_G(u', v') = 1$ to $d_H(u', v') = 2$, we will have that $d_H(u, v) \leq d_G(u,v) + k$. On the other hand if there is a BFS tree center $a$ within distance $3$ of $u'$, then by the triangle inequality
	\begin{eqn}
	    d_H(u,v) &\leq d_H(u,a) + d_H(a,v) \\
	    &\leq d_G(u,u') + d_G(u',a) + d_G(a,u') + d_G(u',v) \\
	    &\leq d_G(u,u') + d_G(u',v) + 3 + 3 = d_G(u,v) + 6
	\end{eqn}
	and since we may assume that $k \geq 6$, we will again have that $d_H(u, v) \leq d_G(u,v) + k$.

	Let $\mathcal{E}_2$ denote the event that $\mathcal{R}_2$ samples a vertex $u_e$ in $N_1(e)$ for every missing edge $e$. Furthermore, let $\mathcal{E}_3$ denote the event that $\mathcal{R}_1$ samples at least one vertex $v_e$ from $N_2(u_e)$ for every edge $e$ for which  $|N_2(u_e)| \geq n/k$. If $\mathcal{E}_2$ and $\mathcal{E}_3$ both hold, then:
	\begin{itemize}
		\item \textbf{Case a}: for all $e \in D$, we have $N_2(u_e) \leq n/k$

		By $\mathcal{E}_2$, there is a truncated BFS center in $N_1(e)$ for all $e \in D$ that reaches both endpoints of $e$. So all missing edges have $2$-hop paths.

		\item \textbf{Case b}: For some $e \in D$, we have $N_2(u_e) > n/k$

		By $\mathcal{E}_3$, there is a full BFS center $v_e$ in $N_2(u_e)$, which is at a distance at most
		\begin{equation}
		    d_G(P,v_e)\leq d_G(P,u_e) + d_G(u_e,v_e)\leq 1 + 2 = 3
		\end{equation}
		of $P$.
	\end{itemize}
	By the arguments above, this sub-case analysis implies that $H$ is a $+s$ spanner for all $u,v$ for which at most $s$ edges are missing from $P$.
\end{itemize}

It remains to show that $\mathcal{E}_1$, $\mathcal{E}_2$ and $\mathcal{E}_3$ hold simultaneously with probability $1 - o(1)$. All three can be made to hold individually with probability $1 - o(1)$ by applying Lemma \ref{lemma:samplecover}, and this will determine the $\tilde{O}$ factors in \cref{swithout}. By a union bound all three events hold simultaneously with probability $1 - o(1)$.

We now consider the communication complexity.
Identifying the vertices of degree at most $\sqrt{sn/k}$ and communicating their incident edges in line \ref{line:sampleline} requires $\tilde{O}(\sqrt{s/k}n^{3/2})$ communication. By Lemma \ref{lemma:bfs_communication} the full BFS trees constructed in line \ref{line:fullBFS} require $\tilde{O}(\sqrt{s/k}n^{3/2} + snk)$ communication. Similarly, the truncated BFS trees found in step 8 require $\tilde{O}(sn/k)$ communication each, for a total of $\tilde{O}(\sqrt{s/k}n^{3/2})$. Adding, we obtain an upper bound of $\tilde{O}(\sqrt{s/k}n^{3/2} + snk)$.
\end{proof}

\section{Multiplicative Spanners}
In this section we study how to compute multiplicative spanners of graphs in the message-passing model. Recall the definition of multiplicative spanners.

\begin{definition}[Multiplicative spanners]
	Given a graph $G$, a subgraph $H$ is a \emph{multiplicative $\alpha$-spanner for $G$} if for all vertex pairs $u, v \in V$,
	$d_{H}(u,v) \leq \alpha \cdot d_{G}(u,v).$
where $d_{G}(u,v)$ and $d_{H}(u,v)$ are the shortest path distances in $G$ and $H$ respectively.
\end{definition}

\subsection{Multiplicative \texorpdfstring{$(2k-1)$}{(2k-1)}-Spanners with Duplication}

We start with a warmup theorem. Consider the case $\alpha = 2$ and with edge duplication.

\begin{theorem} \label{thm:3.2}
The communication cost of computing a multiplicative $\alpha  = 2$-spanner with duplication is $\tilde{\Theta} (s n^2)$.
\end{theorem}
\begin{proof}
For the lower bound, let $K_{n/2, n/2}$ be the complete bipartite graph on $n$ vertices with $n^{2}/4$ edges. Removing any edge $u$,$v$ from this graph increases the distance from $u$ to $v$ to at least $3$, and thus the only multiplicative-$2$ spanner for $K_{n/2, n/2}$ is $K_{n/2, n/2}$ itself. By Lemma \ref{lemma:propertysubgraph}, the communication cost of the multiplicative $2$-spanner problem in the message-passing model is $\Omega(s \cdot n^{2})$.

If all players send their edges to the coordinator, we obtain a matching $\tilde{O}(s n^{2})$ bit communication protocol.
\end{proof}

One can extend the above argument for general multiplicative $(2k-1)$-spanners.

\begin{theorem} \label{thm:mult-with-dup}
The communication cost of the multiplicative $(2k-1)$-spanner problem with edge duplication is $\tilde{O}(s n^{1 + 1/ k})$. Under Erd\H{o}s' girth conjecture \cite{erdos1964extremal}, the bound is tight, in other words the cost is $\tilde{\Theta}(s n^{1 + 1/ k})$.
\end{theorem}

\begin{proof}
The upper bound follows from implementing the well known greedy algorithm for multiplicative spanners, while the lower bound comes from the extremal graphs of large girth given by the girth conjecture. The details are deferred to \cref{section:greedy}.
\end{proof}

\subsection{Multiplicative \texorpdfstring{$(2k-1)$}{(2k-1)}-Spanners without Duplication}

For $k=2$, the additive $2$-spanner algorithm of Theorem \ref{plus2without} immediately gives us a multiplicative $(2k-1)=3$-spanner algorithm with $\tilde O(\sqrt sn^{3/2})$ communication. We show that this bound is in fact tight. We will use the following fact about bipartite biregular graphs follows from Theorem 2 of \cite{YuanshengL03} by taking an appropriate subgraph of their construction:
\begin{corollary}\label{cor:bip-bir-gir6}
Let $s,n$ be such that $\sqrt{sn}$ be a prime power and $\sqrt{n/s}$ is an integer. Then, there exists a bipartite biregular graph of girth $6$ on $\Theta(n)$ vertices where one side has size $\Theta(n/s)$ with common degree $\sqrt{sn}$ and one side with size $\Theta(n)$ with common degree $\sqrt{n/s}$.
\end{corollary}

Using this extremal graph, we obtain the following theorem:
\begin{theorem}\label{thm:mult-3-lb}
The randomized communication cost of the multiplicative $3$-spanner problem without edge duplication is $\Omega(\sqrt{s}n^{3/2})$.
\end{theorem}
\begin{proof}
Recall the graph $Z$ from \cref{cor:bip-bir-gir6} and let $U$ be the partite set with $\Theta(n/s)$ vertices and common degree $\sqrt{sn}$ and let $V$ be the partite set with $\Theta(n)$ vertices and common degree $\sqrt{n/s}$. Note that this graph has $m\coloneqq \Theta(n^{3/2}/\sqrt{s})$. We will reduce $s$ player set disjointness on $m$ elements to the problem of finding multiplicative $3$-spanners without edge duplication.

Consider $s$ copies of the vertex sets $U^1, U^2,\dots, U^s$, each belonging to each of the $s$ players, as well as one copy of the vertex set $V$ belonging to the coordinator. Now given an instance of set disjointness with $s$ players each holding a set $X_i\subseteq[m]$, we define our input graph $G$ by giving the $i$th player the edge indexed by $j\in[m]$ if and only if $j\notin X_i$. That is, if $\{a,b\}\in Z$ is the edge indexed by $j$, then we give $P^i$ the edge $\{(a,i), b\}$, where $(a,i)\in U^i$ is the copy of the vertex $a\in U$ and $b$ is the single copy of the vertex $b\in V$ that belongs to the coordinator. Note that this graph consists of $\Theta(n)$ vertices for the one copy of $V$ and $\Theta(s\cdot n/s) = \Theta(n)$ for the $s$ copies of $U$, and $\Theta(\sqrt{s}n^{3/2})$ edges without duplication.
\begin{figure}
    \centering
    \begin{tikzpicture}
        \node[shape=circle,draw=black,inner sep=2] (U1a) at (0,3) {};
        \node[shape=circle,draw=black,inner sep=2] (U1b) at (0,2.75) {};
        \node[shape=circle,draw=black,inner sep=2] (U1c) at (0,2.5) {};

        \node[shape=circle,draw=black,inner sep=2] (U2a) at (0,2) {};
        \node[shape=circle,draw=black,inner sep=2] (U2b) at (0,1.75) {};
        \node[shape=circle,draw=black,inner sep=2] (U2c) at (0,1.5) {};

        \node[] at (0,1) {$\vdots$};

        \node[shape=circle,draw=black,inner sep=2] (U3a) at (0,0.25) {};
        \node[shape=circle,draw=black,inner sep=2] (U3b) at (0,0) {};
        \node[shape=circle,draw=black,inner sep=2] (U3c) at (0,-0.25) {};

        \node[shape=circle,draw=black,inner sep=2] (Va) at (5,1.75) {};
        \node[shape=circle,draw=black,inner sep=2] (Vb) at (5,1.5) {};
        \node[shape=circle,draw=black,inner sep=2] (Vc) at (5,1.25) {};
        \node[shape=circle,draw=black,inner sep=2] (Vd) at (5,1) {};
        \node[shape=circle,draw=black,inner sep=2] (Ve) at (5,0.75) {};
        \node[shape=circle,draw=black,inner sep=2] (Vf) at (5,0.5) {};

        \path (U1a) edge[thick,red] (Vc);
        \path (U1b) edge[thick,red] (Va);
        \path (U1b) edge[thick,red] (Ve);
        \path (U2a) edge[thick,blue] (Vf);
        \path (U2b) edge[thick,blue] (Vf);
        \path (U2c) edge[thick,blue] (Vb);
        \path (U3b) edge[thick,green] (Vd);
        \path (U3c) edge[thick,green] (Ve);
        \path (U3c) edge[thick,green] (Va);

        \node[anchor=east] at (-0.1,2.75) {$U^1$};
        \node[anchor=east] at (-0.1,1.75) {$U^2$};
        \node[anchor=east] at (-0.1,0) {$U^s$};
        \node[anchor=west] at (5.1,1.125) {$V$};
    \end{tikzpicture}
    \caption{Our hard input instance.}
    \label{fig:bipartite-hard-instance}
\end{figure}
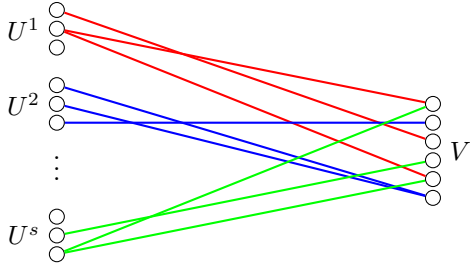

We now show that the $s$ sets $X_i$ simultaneously intersect if and only if there is an edge $\{a,b\}\in Z$ that is missing from all $s$ copies of the graph $Z$ in the spanner $H$. It's clear that if $j\in\bigcap_{i=1}^s X_i$, then the edge $\{a,b\}\in Z$ indexed by $j$ cannot be in the spanner $H$. Now suppose that no copy of the edge $\{a,b\}\in Z$ indexed by $j$ is in the spanner $H$ and suppose for contradiction that $j\notin X_i$ for some $i$, that is, the edge $\{a,b\}\in G$ belonged to some player $P^i$. We claim that there is no path of length at most $3$ in the spanner $H$, contradicting that $H$ is a multiplicative $3$-spanner. Indeed, suppose for contradiction that a path $(a,i),v_1,(v_2,i'),b$ were in the spanner. Then $v_1\neq b$ since we assumed that there were no copies of $\{a,b\}$ in the spanner, and similarly, $v_2\neq a$. Then since $\{a,b\}$ is also an edge of $Z$, this implies that there is a $4$-cycle $a,v_1,v_2,b$ in $Z$, which contradicts that $Z$ is of girth $6$.

We thus conclude that the randomized communication complexity of this problem is
\begin{eqn}
    \Omega(sm) = \Omega\parens*{s\frac{n^{3/2}}{\sqrt{s}}} = \Omega(\sqrt{s}n^{3/2}),
\end{eqn}
as desired.
\end{proof}

\begin{remark}
For bipartite biregular graphs of larger girth, the optimal size of the graph is unknown. However, a simple counting argument known as the \emph{Moore bound} gives a lower bound on the number of vertices of a bipartite biregular graph of prescribed bidegree and girth \cite{Hoory02, FilipovskiRJ19, araujo2019bipartite}, which shows limitations of the above technique for proving communication lower bounds for multiplicative spanners of larger distortion. More specifically, let $g = 2k+2$ be the girth and let the two degrees be $\{d,sd\}$ (as we require one side of the bipartite graph to be of size $\Theta(n/s)$ in our proof technique). Then the Moore bound states that when $k$ is odd, then the number of vertices is at least
\begin{eqn}
    n = \Omega\parens*{(sd)^{(k+1)/2}d^{(k-1)/2}} = \Omega(s^{(k+1)/2}d^{k})
\end{eqn}
which implies that we can't get a bound better than $\Omega(sdn) = \Omega(s^{1/2-1/2k}n^{1+1/k})$. We will in fact be able to show this bound under the Erd\H{o}s girth conjecture for all $k$, as we show next. On the other hand, when $k$ is even, then the Moore bound is
\begin{eqn}
    n = \Omega((sd)^{k/2}d^{k/2}) = \Omega(s^{k/2}d^k)
\end{eqn}
which implies that the best we can do is $\Omega(sdn) = \Omega(s^{1/2}n^{1+1/k})$, which is slightly better than the previous bound. However, it is known that these Moore bounds are not tight everywhere, and counterexamples exist in some limited parameter regimes, e.g.\ Theorem 4 of \cite{de1997dense}.

Finally, we note that more robust versions of the Moore bound have been shown for bipartite biregular graphs \cite{Hoory02}, where an analogue of the above bound holds even for irregular graphs.
\end{remark}

For $k\geq 3$, one can implement the multiplicative spanner algorithm of \cite{baswana2007simple} to get asymptotically better dependence on the number of servers $s$ than the lower bound of Theorem \ref{thm:mult-with-dup}. This separates the with edge duplication model from the without duplication model for all $k$, given the lower bound of \cref{thm:mult-with-dup}.

\begin{theorem}\label{thm:mult-without-dup-ub}
For $k\geq 3$, the communication cost of the multiplicative $(2k-1)$-spanner problem without edge duplication is $\tilde O(ks^{1-2/k}n^{1+1/k} + snk)$.
\end{theorem}
\begin{proof}
The result is obtained just by balancing parameters in the cluster-cluster joining algorithm of \cite{baswana2007simple}. The details are deferred to \cref{section:bs-mult-spanner}.
\end{proof}
\begin{remark}
If we instead implement the vertex-cluster joining version of the algorithm of \cite{baswana2007simple}, then the bound weakens to $\tilde O(ks^{1-1/k}n^{1+1/k}+snk)$
\end{remark}

We also show that in this model, a polynomial dependence on the parameter $s$ is necessary. The lower bound matches the algorithm of Theorem \ref{thm:mult-without-dup-ub} exactly for $k=3$ up to polylog factors, giving a communication complexity of $\tilde\Theta(s^{1/3}n^{4/3})$ in this case. For general $k$, the bounds are off by a factor of $\tilde O(s^{1/2-3/2k})$. Interestingly, this technique is not able to get us tight results for $k=2$, giving a lower bound of $\Omega(s^{1/4}n^{3/2})$ instead.

\begin{theorem} \label{thm:mult-without-dup-lb}
Under Erd\H{o}s' girth conjecture, the randomized communication cost of the multiplicative $(2k-1)$-spanner problem without edge duplication is $\Omega(s^{1/2 - 1/2k}n ^ {1+ 1/k} + sn)$.
\end{theorem}
\begin{proof}
We prove the $\Omega(s^{1/2 - 1/2k}n ^ {1+ 1/k})$ lower bound via a reduction from the lower bound for the multiplicative $(2k-1)$-spanner problem \emph{with} duplication.

Consider an instance $G$ of the multiplicative $(2k-1)$-spanner problem \emph{with} duplication on $s$ servers and $n/s^{1/2}$ vertices. We then construct an instance of multiplicative $(2k-1)$-spanner problem \emph{without} duplication on $s$ servers and $n$ vertices as follows.

We first construct a product graph $G'$ on $n$ vertices by replacing every vertex $v$ in $G$ with a set $S_v$ of $s^{1/2}$ vertices. Then, for a pair of vertices $\{u,v\}$ in the original graph $G$, there are $s$ distinct edges between the two corresponding groups of vertices $\{S_u,S_v\}$. Now note that there are at most $s$ copies of each edge $\{u,v\}$ in the original graph $G$ across all the servers, so we can deterministically assign each server's copy of an edge $\{u,v\}\in E(G)$ to a distinct edge $\{u',v'\}\in E(G')$ for $u'\in S_u$ and $v'\in S_v$ without any additional communication (Figure \ref{fig:dup-to-no-dup}).
\begin{figure}
    \centering
    \begin{tikzpicture}
        \node[anchor=south] at (0,0.3) {$u$};
        \node[anchor=south] at (3,1.3) {$v$};
        \node[shape=circle,draw=black] (A) at (0,0) {};
        \node[shape=circle,draw=black] (B) at (3,1) {};
        \path (A) edge[thick,blue] (B);
        \path (A) edge[thick,red,bend left] (B);
        \path (A) edge[thick,green,bend right] (B);

        \node[anchor=south] at (5,0.5) {$S_u$};
        \node[anchor=south] at (8,1.5) {$S_v$};
        \node[shape=circle,draw=black] (A1) at (5,-0.3) {};
        \node[shape=circle,draw=black] (A2) at (5,0.3) {};
        \node[shape=circle,draw=black] (A3) at (4.7,0) {};
        \node[shape=circle,draw=black] (A4) at (5.3,0) {};
        \node[shape=circle,draw=black] (B1) at (8,1.3) {};
        \node[shape=circle,draw=black] (B2) at (8,0.7) {};
        \node[shape=circle,draw=black] (B3) at (8.3,1) {};
        \node[shape=circle,draw=black] (B4) at (7.7,1) {};
        \path (A4) edge[thick,blue] (B4);
        \path (A2) edge[thick,red,bend left] (B1);
        \path (A1) edge[thick,green,bend right] (B2);
    \end{tikzpicture}

    \caption{Converting an instance with edge duplication to one without.}
    \label{fig:dup-to-no-dup}
\end{figure}
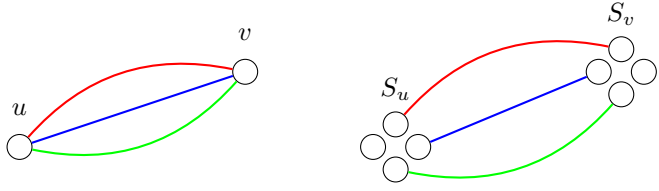
This is now an instance of multiplicative $(2k-1)$-spanner problem \emph{without} duplication on $s$ servers and $n$ vertices as follows, so we can run any algorithm $\mathcal A$ on this instance so that the coordinator ends up with a multiplicative $(2k-1)$-spanner $H'$ of $G'$. Finally, the coordinator constructs a multiplicative $(2k-1)$-spanner $H$ of $G$ by including an edge $\{u,v\}\in E(G)$ in $H$ if and only if there is an edge between $S_u$ and $S_v$ in $H'$.

We show that this is a correct protocol for computing a multiplicative $(2k-1)$-spanner $H$ of $G$. Consider an edge $\{u,v\}\in E(G)$. Then, there is an edge $\{u',v'\}$ between $S_u$ and $S_v$ in $G'$ by construction. Thus, there exists a path of length at most $(2k-1)$ in $H'$ between $u'$ and $v'$. Now every edge in this path is an edge between two groups of vertices $S_{w_1}$ and $S_{w_2}$, in which case we have included the edge $\{w_1,w_2\}$ in $H$. Thus, we have included a path of length at most $(2k-1)$ between the vertices $u$ and $v$ in $H$.

By the lower bound of Theorem \ref{thm:mult-with-dup}, $\mathcal A$ requires
\begin{eqn}
    \Omega\parens*{s\parens*{\frac{n}{s^{1/2}}}^{1+\frac1k}} = \Omega\parens*{s^{\frac12 - \frac1{2k}}n^{1+\frac1k}}
\end{eqn}
communication, as desired.

As before, the $\Omega(sn)$ communication lower bound for graph connectivity from \cite{woodruffzhang13} implies that this bound can be strengthened to $\Omega(s^{1/2-1/2k}n^{1+1/k} + sn)$.
\end{proof}
%

\subsection{Simultaneous Communication of Multiplicative \texorpdfstring{$(2k-1)$}{(2k-1)}-Spanners}
We now prove our results for simultaneous communication for multiplicative $(2k-1)$-spanners.

Our algorithm comes from observing that for multiplicative spanners, each server can just locally compute a multiplicative $(2k-1)$-spanner of size $O(n^{1+1/k})$ and send it to the server for a $\tilde O(sn^{1+1/k})$ communication algorithm, which turns out to be optimal.

\begin{theorem}
The deterministic simultaneous communication complexity of multiplicative $(2k-1)$-spanners problem with duplication is $\tilde O(sn^{1+1/k})$.
\end{theorem}

To prove the lower bound, we will make use of the following lemma.

\begin{lemma}\label{lem:disj-girth-graphs}
Let $s = o(n^{1/3-1/3k})$. Then under the Erd\H{o}s girth conjecture, there exist $s$ pairwise edge-disjoint graphs $E_1,E_2,\dots,E_s$ on $n$ vertices and $\Theta(n^{1+1/k})$ edges, each of girth $2k+2$.
\end{lemma}
\begin{proof}
Under the Erd\H{o}s girth conjecture, there exists a graph $G$ on $n$ vertices with $\Theta(n^{1+1/k})$ edges and girth $2k+2$. We now choose the $s$ pairwise edge-disjoint graphs on $n$ vertices as follows. First draw a random permutation of $G$ for each server $P^j$ for $j\in[s]$ by drawing a random permutation $\pi_j:[n]\to[n]$ of the vertices and giving the server $P^j$ the edge set $\braces*{\{\pi_j(u),\pi_j(v)\} : \{u,v\}\in E(G)}$. To get pairwise edge-disjoint graphs, we now just delete any shared edges to produce our final edges $E_j$ for $j\in[s]$.

Note that any subgraph of $G$ also has girth at least $2k-1$, so $E_j$ has girth at least $2k+2$ for all $j\in[s]$. It remains to show that in our parameter regime, this yields graphs of the desired size. Fix two distinct players $P^i,P^j$ and edges $e_1\in E_i$ and $e_2\in E_j$. Then, the probability that $e_1$ collides with $e_2$ is $(n-2)!/n! = 1/n(n-1)$ and thus the expected number of edges shared between $P^i$ and $P^j$ is $n^{1+1/k}n^{1+1/k}/n(n-1) = \Theta(n^{2/k})$. By Markov's inequality, we delete $O(s^2n^{2/k})$ edges between these two players with probability at least $1-O(s^{-2})$. By the union bound, this is true simultaneously for all pairs of players with positive probability. In this event, each player deleted at most $s\cdot O(s^2n^{2/k}) = O(s^3n^{2/k}) = o(n^{1-1/k}n^{2/k}) = o(n^{1+1/k})$ edges, so each player still has a graph of size $\Theta(n^{1+1/k})$.
\end{proof}

Our lower bound, intuitively, now comes from either giving each of the players the above graphs with probability $1/2$, or giving only one player one of the above graphs with probability $1/2$, so that everyone always has to send their entire graph.

We first show the following tight lower bound for the communication complexity of multiplicative spanners for two players, even for the problem of just computing a multiplicative $(2k-1)$-approximation of all pairwise distances. That is, we consider the weaker version of the problem where, given a graph $G$, we only require the coordinator to put an oracle $O_G$ that, when queried for two vertices $v$ and $w$, returns $O_G(v,w)$ such that $d_G(v,w)\leq O(v,w)\leq (2k-1)d_G(v,w)$.

\begin{theorem}\label{thm:2-player-mult}
Suppose there exists a graph $G_n$ on $n$ vertices with girth $2k+2$ and $m = \Omega(n^{1+1/k})$ edges. Then, the randomized communication cost of the multiplicative $(2k-1)$-approximate distance oracle problem without edge duplication is $\Omega(n ^ {1+ 1/k})$.
\end{theorem}
\begin{proof}
The proof follows easily from the observation of \cite{tz05} that $(2k-1)$-approximate distance oracles uniquely determine subgraphs of graphs of girth $2k+2$, along with standard arguments. The details are given in \cref{section:dist-oracle-lb}.
\end{proof}

We now amplify this two player result to the $s$ player simultaneous communication result in the following theorem.

\begin{theorem}
The randomized simultaneous communication complexity of multiplicative $(2k-1)$-approximate distance oracle problem without duplication is $\Omega(sn^{1+1/k})$.
\end{theorem}
\begin{proof}
By the distributed version of Yao's minimax principle, Lemma 1 of \cite{woodruffzhang13}, it suffices to argue that there exists an input distribution such that any deterministic algorithm $\mathcal A$ requires an expected $\Omega(sn^{1+1/k})$ bits of simultaneous communication in order to succeed with probability at least $11/12$ over the randomness of the input distribution. Recall the graphs $E_1,E_2,\dots,E_s$ of \cref{lem:disj-girth-graphs}, each of size $m \coloneqq \Theta(n^{1+1/k})$.

Recall that by Yao's minimax principle, the worst case randomized cost is at least the expected cost of a deterministic algorithm over an input distribution. For each girth graph $E_j$, let $\mu(E_j)$ be an input distribution that witnesses this expected cost for any deterministic algorithm. Let $\mu_1$ be the input distribution that draws the input for player $P^j$ from $\mu(E_j)$ for every player $j\in[s]$. Let $\mu_2$ be the input distribution that draws a uniformly random index $J_*\sim[s]$ and gives $P^{J_*}$ a subset $H\subseteq E_{J_*}$ drawn from $\mu(E_{J_*})$ and everyone else the empty graph $\varnothing$. Now let $Z$ be a uniformly random bit. Then, we choose our input distribution to draw from $\mu_1$ if $Z = 0$ and $\mu_2$ if $Z = 1$.

Note that
\begin{eqn}
    \Pr_\mu\parens*{\text{$\mathcal A$ fails}} = \frac12\parens*{\Pr_\mu\parens*{\text{$\mathcal A$ fails}\mid Z = 0} + \Pr_\mu\parens*{\text{$\mathcal A$ fails}\mid Z = 1}}\leq \frac1{12}
\end{eqn}
so
\begin{eqn}
    \Pr_\mu\parens*{\text{$\mathcal A$ fails}\mid Z = 1} = \Pr_{\mu_2}\parens*{\text{$\mathcal A$ fails}} = \frac1s\sum_{j=1}^s \Pr_{\mu_2}\parens*{\text{$\mathcal A$ fails}\mid J_* = j} \leq \frac16
\end{eqn}
and thus by averaging, for at least half the indices $j\in[s]$, say the indices $S\subseteq[s]$, we have that the failure rate is $\Pr_{\mu_2}\parens*{\text{$\mathcal A$ fails}\mid J_* = j}\leq 1/3$. Now note that \cref{thm:2-player-mult} trivially implies that the one-way communication complexity of the multiplicative $(2k-1)$-distance oracle problem is $\Omega(n^{1+1/k})$ as well. Then these players $j\in S$ send an expected $\Omega(n^{1+1/k})$ over their respective input distributions $\mu(E_j)$. Thus, denoting the random variable for the communication of player $j$ by $C_j$, the required expected simultaneous communication is
\begin{eqn}
    \E_\mu\parens*{\sum_{j=1}^s C_j} &= \frac12\sum_{i=1}^s\E_\mu\parens*{C_j\mid Z=0} + \E_\mu\parens*{C_j\mid Z=1} \geq \frac12\sum_{j\in S}\E_\mu\parens*{C_j\mid Z=0} \\
    &= \frac12\sum_{j\in S}\E_{\mu_1}\parens*{C_j} = \frac12\sum_{j\in S}\E_{\mu(E_j)}\parens*{C_j} \geq \abs*{S}\Omega(n^{1+1/k}) = \Omega(sn^{1+1/k})
\end{eqn}
as desired.

\end{proof}

\subsection{Multiplicative \texorpdfstring{$(2k-1)$}{(2k-1)}-Spanners in the Dynamic Streaming Model}
Finally, we note that implementing the Baswana-Sen cluster-cluster joining algorithm \cite{baswana2007simple} in the turnstile streaming model gives a $(\floor{k/2}+1)$-pass algorithm.

\begin{theorem}\label{thm:spanner-dynamic}
There exists an algorithm for constructing a multiplicative $(2k-1)$-spanner using $\tilde O(n^{1+1/k})$ space and $\floor{k/2}+1$ passes in the dynamic streaming model.
\end{theorem}
\begin{proof}
We defer the details to \cref{section:spanner-dynamic}.
\end{proof}

The space-distortion tradeoff here is optimal under the Erd\H{o}s girth conjecture, as graphs given by this conjecture must output themselves as spanners, which takes $\Omega(n^{1+1/k})$ bits of space.

\section{Conclusions}
We initiated the study of communication versus spanner quality in the message-passing model of communication, in which the edges of a graph are arbitrarily distributed, with or without duplication, across two or more players, and the players wish to execute a low
communication protocol to compute a spanner. We believe there are several surprising aspects of these problems illustrated by our work, illustrating separations between models with and without edge duplication.

One open question is whether it is possible to obtain an additive spanner with constant distortion with $O(n^{4/3})$ communication for constant $s$. We show it is possible to obtain $O(n^{3/2})$ communication and constant distortion, but in the non-distributed setting it is possible to obtain an additive $6$-spanner with $O(n^{4/3})$ edges. Since known constructions involve computing many partial breadth-first search trees, we are not able to implement them in the message-passing model, nor are we able to exploit any of the literature for computing distributed BFS trees (see, e.g., \cite{a89}), without spending $\Omega(n^2)$ communication in the message-passing model. Yet another question is to extend our techniques to other notions of spanners, such as distance preservers \cite{coppersmith2006sparse} or mixed additive and multiplicative spanners \cite{ep04}; see also the $(k, k-1)$ spanners in \cite{BKMP05}.

\section{Acknowledgements}
We would like to thank Gregory Kehne, Roie Levin, Chen Shao and Srikanta Tirthapura for helpful discussions, as well as the anonymous reviewers for their useful feedback. D. Woodruff would like to acknowledge support from the Office of Naval Research (ONR) grant N00014-18-1-2562.

\bibliographystyle{alpha}
\bibliography{spanner_bib}

\appendix
\section{Proofs of Simple Lemmas}\label{section:simple-lemmas}
\begin{proof}[Proof of \cref{lemma:samplecover}]
For a fixed set $C \in \mathcal{C}$, the probability we do not sample any elements from $C$ is no more than
$ (1 - \frac{t}{|\mathcal{U}|})^ {\frac{|\mathcal{U}|}{t} \log |\mathcal{C} / \delta|}
\leq  e^{- k \log |\mathcal{C}|} =  \frac{\delta}{|\mathcal{C}|} $.
Taking a union bound over all sets in $\mathcal{C}$ yields the claim.
\end{proof}

\begin{proof}[Proof of \cref{lemma:bfs_communication}]
	We give a simple distributed protocol for BFS. The coordinator maintains a partial BFS tree through the algorithm and an active set of vertices, both initialized to be the starting vertex $x$. Repeat the following until the active set is empty. Every round, the coordinator broadcasts the current active set $A_t$ to all players, and subsequently removes them from the active set. Each player responds by sending the coordinator any neighbors of $A_t$ that have never been in the active set, along with some edge from each neighbor to $A_t$. The coordinator adds the new neighbors to the BFS tree appropriately (breaking ties between new edges arbitrarily), and adds the new neighbors to the next round's active set $A_{t+1}$.

	Each vertex in the graph is broadcast at most once by the coordinator as an active vertex to each player, and is sent by each player at most once to the coordinator (along with its accompanying edge). Thus the total communication of this protocol is $\tilde O(sn)$.
\end{proof}

\section{Additive \texorpdfstring{$k$}{k}-Spanner Lower Bounds}\label{section:additiveNoDupLB}
We give the details for the proofs of \cref{thm:additive-lb} (with duplication model) and \cref{thm:additiveNoDupLB} (without duplication model).

\begin{proof}[Proof of \cref{thm:additive-lb}]
Recall the extremal graph construction $G$ of \cite{abboud20164} that shows that additive $k$-spanners must have size $\Omega(n^{4/3-o(1)})$. This graph is constructed so that there are $m\coloneqq O(n^{4/3-o(1)})$ pairs of vertices $\{s,t\}$, each associated with a set of edges $C^{s,t}$ known as \emph{clique edge sets}, such that distinct pairs of vertices have disjoint clique edge sets (Claim 3 of \cite{abboud20164}), and every path between these pairs of vertices $\{s,t\}$ with addition distortion at most $k$ must include some edge from $C^{s,t}$ (Claim 5 of \cite{abboud20164}. Let $P$ denote the set of these special pairs of vertices.

We now use this construction to solve an instance of $s$-player set disjointness on $m$ elements using an algorithm for additive $k$-spanners in a similar way as \cref{lemma:propertysubgraph} as follows. We first give any non-clique edge set edge to the coordinator. Now suppose player $i$ is given the input set $X_i\subseteq[m]$. Then, we give player $i$ the entire clique edge set corresponding to the $j$th pair for $j\in[m]$ if and only if $j\notin X_i$. Note that for any given pair $\{s,t\}\in P$, if $C^{s,t}$ was given to player $i$, then an additive $k$-spanner must include some edge of $C^{s,t}$ as mentioned before, and the additive $k$-spanner will include no edge of $C^{s,t}$ otherwise. Thus, $\bigcap_i X_i = \varnothing$ if and only if the additive $k$-spanner output by the coordinator is the entire $G$ itself. Thus, computing the additive $k$-spanner requires $\Omega(sm) = \Omega(sn^{4/3-o(1)})$, as desired.
\end{proof}

\begin{proof}[Proof of \cref{thm:additiveNoDupLB}]
It was shown in the proof of Theorem 2 in \cite{abboud20164} that for any constant $k$ and any $\epsilon > 0$, there is a family of graphs $\mathcal{S}_n$ on $n$ vertices of size $2^{\Omega(n^{4/3} - \epsilon)}$ such that for any two distinct graphs $G_1$, $G_2$ in this family, there is some vertex pair $x,y$ such that $|d_{G_1}(x,y) - d_{G_2}(x,y)| \geq k$.

	Consider the distribution that samples a graph from $\mathcal{S}_n$ uniformly at random, assigns all its edges to non-coordinator player $A$, and assigns nothing to the coordinator $B$. Consider any deterministic distributed protocol for computing multiplicative additive $k$-spanners that requires $o(n^{4/3 - \epsilon})$ communication. Since the messages sent by $A$ in this protocol must uniquely identify which graph in $\mathcal{S}_n$ was sampled, at most $2^{o(n^{4/3 - \epsilon})}$ of the possible input graphs in $\mathcal{S}_n$ will produce transcripts from which $B$ can construct a valid spanner. Thus this protocol must fail with probability $1 - o(1)$. Lemma 1 of \cite{woodruffzhang13} (a distributed version of Yao's Lemma \cite{yao1977probabilistic}) implies that the randomized communication of this problem is $\Omega(n^{4/3 - \epsilon})$.

	It was also shown in \cite{woodruffzhang13} that deciding graph connectivity in the message passing model requires $\Omega(sn)$ bits. Thus, the communication of this problem has a lower bound of $\Omega(n^{4/3 - o(1)} + sn)$.
%
%
\end{proof}

\section{Greedy Algorithm Multiplicative Spanners in the Message-Passing Model}\label{section:greedy}
We give the details for the proof of \cref{thm:mult-with-dup}.
\begin{proof}
	For the upper bound, consider the greedy Algorithm \ref{alg:multkspanner}.
	\begin{algorithm}
		\caption{$\times(2k-1)$-spanner with edge duplication}
		\label{alg:multkspanner}
		\begin{algorithmic}[1]
			\Require $G = (V, E)$.
			\Ensure $H$, $(2k-1)$-spanner of $G$.
			\State Initialize $H = (V,F), F = \varnothing$.
			\For {$i$ in $[n]$}
			\For {$e \in E_i$}
			\If {$(V, F \cup \{e\})$ does not contain a cycle of length less than or equal to $2k$}
			\State $F = F \cup \{e\}$.
			\EndIf
			\EndFor
			\EndFor
			\State \Return $H = (V, F)$.
		\end{algorithmic}
	\end{algorithm}

	First we note that the algorithm will produce a $(2k-1)$-spanner. For each edge $(x,y) \in E$, if $(x,y) \notin F$, then $d_F(x,y) \leq 2k-1$ since including the edge $(x,y)$ would close a cycle of length $\leq 2k$. Thus the output $H$ is a $(2k-1)$-spanner.

	Next we argue that the algorithm can be implemented in the message-passing model with $O(s n^{1 + 1/ k})$ bits of communication. Each player in order decides which of its edges to include in the current version of $F$, then forwards the updated $F$ to the next player. By construction the graph produced by the algorithm has girth greater than $2k$. It is well known that graphs with girth greater than $2k$ have $O(n^{1 +  1/ k})$ edges (see e.g. \cite{ADDJS93}). Thus $F$ never has more than $O(n^{1 + 1/k})$ edges and the total amount of communication required is $\tilde{O}(s \cdot n ^ {1+ 1/k})$.

	For the lower bound, under the girth conjecture there is a family of graphs $G_n$ on $n$ vertices with girth $2k+1$ and $\tilde{\Omega}(n ^ {1 + 1/k})$ edges. Since the only multiplicative $k$-spanner of $G_n$ is $G_n$ itself, the lower bound follows from Lemma \ref{lemma:propertysubgraph}.
\end{proof}

\section{Baswana-Sen Multiplicative Spanner in the Message-Passing Model}\label{section:bs-mult-spanner}
We give the details for Theorem \ref{thm:mult-without-dup-ub}.

\begin{proof}[Proof of \cref{thm:mult-without-dup-ub}]
Following the cluster-cluster joining algorithm of \cite{baswana2007simple}, we give two slightly different algorithms depending on whether $k$ is odd or even.

\paragraph{Algorithm for $k$ odd.}
Let $k = 2\ell + 1$ for $\ell\geq 1$.
\begin{itemize}
    \item \textbf{Phase 1: Initializing clusters}

    We first include all edges incident to vertices of degree at most $d_1 = s^{1-2/k}n^{1/k}$. Let $\mathcal C_0$ be a sample of vertices drawn independently by the coordinator with probability $\log n/d_1$ each. We will think of these vertices as cluster centers. This can be broadcasted to each of the servers. Now for each vertex, if it is adjacent to a cluster center, each server sends such an edge to the coordinator. This creates a set of $\tilde O(n/d_1)$ clusters, each of radius $1$.

    \item \textbf{Phase 2: Expanding clusters}

    Let $d_2 = n^{1/k}/s^{2/k}$. For $\ell-1$ iterations from $i = 1,\dots,\ell-1$, we expand the clusters by one layer at a time as follows. Sample a set of clusters $\mathcal C_i\subseteq\mathcal C_{i-1}$ independently with probability $\log n/d_2$ each. Now for each vertex, we check if it is adjacent to some sampled cluster $C\in\mathcal C_i$ or not by having each server send a bit for each vertex. If some server indicates that the vertex is adjacent to a cluster, then we can have that server send that edge to the coordinator, and otherwise, we tell every server to add an edge to all adjacent clusters $C\in\mathcal C_{i-1}$. After iteration $i$, each cluster in $\mathcal C_i$ has radius $i+1$.

    \item \textbf{Phase 3: Connecting clusters}

    Finally, after the $\ell-1$ iterations, we add an edge between every pair of clusters in $\mathcal C_{\ell-1}$.
\end{itemize}

We first argue correctness. Let $\{u,v\}$ be a missing edge with $\deg(u), \deg(v) > d_1$. Then, both $u$ and $v$ belong to a cluster in $\mathcal C_0$ after phase 1 with high probability. We now maintain the loop invariant that $u$ and $v$ are either already well-approximated in the spanner or belongs to a cluster. Consider the $i$th iteration in phase 2. If either $u$ or $v$ are not adjacent to any sampled clusters in $\mathcal C_i$, WLOG say $u$, then $u$ is adjacent to $v$'s cluster and thus is connected to it; since $v$'s cluster has radius $i$, this yields a $2i+1\leq 2\ell-1$ factor approximation. Otherwise, both $u$ and $v$ are adjacent to some sampled cluster, and thus gets added to a cluster of radius $i+1$ in $\mathcal C_i$. At the end of the $\ell-1$ iterations, both $u$ and $v$ belong to clusters of radius $\ell$, and these clusters are adjacent since $\{u,v\}$ is an edge in the graph. Thus, we connect them in phase 3. Let $\{u',v'\}$ be this connecting edge, with $u'$ in $u$'s cluster and $v'$ in $v$'s cluster. Then, it takes at most $2\ell$ to get from $u$ to $u'$, $1$ to get from $u'$ to $v'$, and $2\ell$ to get from $v'$ to $v$, which is a total of
\begin{eqn}
    2\ell + 1 + 2\ell = 2k - 1
\end{eqn}
as desired.

We now argue the communication. In phase 1, it takes $d_1 n = \tilde O(s^{1-2/k}n^{1+1/k})$ bits of communication to send all the low degree edges and $sn$ to assign vertices to clusters in $\mathcal C_0$. In phase 2, it takes $\tilde O(sn)$ bits of communication to assign vertices to sampled clusters, and if a vertex is not adjacent to a sampled cluster, then it is adjacent to at most $d_2$ clusters and thus it takes $sd_2 = \tilde O(s^{1-2/k}n^{1/k})$ communication to connect a vertex to the clusters and thus $s^{1-2/k}n^{1/k}$ total. Finally, there are
\begin{eqn}
    \tilde O\parens*{\frac{n}{d_1d_2^{\ell-1}}} = \tilde O\parens*{\frac{n}{sd_2^\ell}} = \tilde O\parens*{\frac{ns^{2\ell/k}}{sn^{\ell/k}}} = \tilde O\parens*{\frac{n^{1 - \ell/k}}{s^{1/k}}}
\end{eqn}
clusters at the end in expectation, so it takes
\begin{eqn}
    \tilde O\parens*{s\parens*{\frac{n^{1 - \ell/k}}{s^{1/k}}}^2} = \tilde O\parens*{s^{1-2/k}n^{1+1/k}}
\end{eqn}
communication to connect them all.

\paragraph{Algorithm for $k$ even.}
Let $k = 2\ell$ for $\ell\geq 2$.
\begin{itemize}
    \item \textbf{Phase 1: Initializing clusters}

    We first include all edges incident to vertices of degree at most $d_1 = s^{1-2/k}n^{1/k}$. Let $\mathcal C_0$ be a sample of vertices drawn independently by the coordinator with probability $\log n/d_1$ each. For each vertex, add an edge to an adjacent cluster center if one exists.

    \item \textbf{Phase 2: Expanding clusters}

    Let $d_2 = n^{1/k}/s^{2/k}$. For $\ell-1$ iterations from $i = 1,\dots,\ell-1$, we expand the clusters by one layer at a time as follows. Sample a set of clusters $\mathcal C_i\subseteq\mathcal C_{i-1}$ independently with probability $\log n/d_2$ each. Now for each vertex, if it is adjacent to a sampled cluster $C\in\mathcal C_i$, then add an edge to it, and otherwise, add an edge to all adjacent clusters $C\in\mathcal C_{i-1}$. After iteration $i$, each cluster in $\mathcal C_i$ has radius $i+1$.

    \item \textbf{Phase 3: Connecting clusters}

    Finally, after the $\ell-1$ iterations, we add an edge between every pair of clusters in $\mathcal C_{\ell-1}$ and $\mathcal C_{\ell-2}$.
\end{itemize}

Correctness and communication for phases 1 and 2 are similar to before, so we just show the communication for phase 3. In expectation, there are $\tilde O\parens*{n/d_1d_2^{\ell-1}}$ clusters in $\mathcal C_{\ell-1}$ and $n/d_1d_2^{\ell-2}$ clusters in $\mathcal C_{\ell-2}$, so it takes
\begin{eqn}
\tilde O\parens*{s\frac{n}{d_1d_2^{\ell-1}}\frac{n}{d_1d_2^{\ell-2}}} = \tilde O\parens*{s\frac{n^2}{s^2d_2^{2\ell-1}}} = \tilde O\parens*{s\frac{n^2s^{2(2\ell-1)/k}}{s^2n^{(2\ell-1)/k}}} = \tilde O\parens*{s^{1-2/k}n^{1+1/k}}
\end{eqn}
communication to connect the pairs.
\end{proof}

\section{Two Player Multiplicative \texorpdfstring{$(2k-1)$}{(2k-1)}-Approximate Distance Oracle Lower Bound}\label{section:dist-oracle-lb}
We give the details for the proof of \cref{thm:2-player-mult}.
\begin{proof}[Proof of \cref{thm:2-player-mult}]
We follow the ideas in the approximate distance oracle space lower bound, Proposition 5.1, of \cite{tz05}. Let $Z\subseteq G_n$ be any subgraph and let $O_Z$ be a multiplicative $(2k-1)$-distance oracle for it. Now consider any edge $\{v,w\}\in G_n$. If $\{v,w\}\in Z$, then $O_Z(v,w)\leq 2k-1$, while if $\{v,w\}\notin Z$, then $O_Z(v,w)\geq 2k+1$ since the girth of $G_n$ is $2k+2$. Thus, a the distance oracle $O_Z$ uniquely determines $Z$.

Now consider the distribution that samples a subgraph $Z$ of $G_n$ uniformly at random, assigns all its edges to a non-coordinator player $A$, and assigns nothing to the coordinator $B$. Consider a randomized protocol that fails with probability $1/3$, with transcript $\Pi(Z)$ of the communication between $A$ and $B$. Note that the probability that the output $O_Z$ fails to identify $Z$ is just the probability that the protocol fails. Then by Fano's inequality,
\begin{eqn}
    H(Z\mid \Pi(Z))\leq h\parens*{\frac13} + \frac13\log(2^m-1)\leq  \frac12m.
\end{eqn}
By Proposition 4.3 of \cite{Bar-YossefJKS02}, the randomized communication complexity is at least the mutual information $I(Z;\Pi(Z))$, which gives a lower bound of
\begin{eqn}
    I(Z;\Pi(Z)) = H(Z) - H(Z\mid \Pi(Z))\geq m - \frac12 m = \frac12 m = \Omega(n^{1+1/k})
\end{eqn}
as desired.
\end{proof}

\section{Baswana-Sen Multiplicative Spanner in the Dynamic Streaming Model}\label{section:spanner-dynamic}
We give the details for \cref{thm:spanner-dynamic}. We will use the $\ell_0$-sampler, a standard primitive in the streaming literature which gives us access to uniform sampling:
\begin{theorem}[\cite{JowhariST11}]
There is an algorithm in the dynamic streaming model for sampling a uniformly random nonzero entry of the underlying vector $x$ that errs with probability at most $\delta$ and uses $O(\log^2 n\log \delta^{-1})$ space.
\end{theorem}

Using this, we have the following:
\begin{proof}[Proof of \cref{thm:spanner-dynamic}]
Following the cluster-cluster joining algorithm of \cite{baswana2007simple}, we give two slightly different algorithms depending on whether $k$ is odd or even.

Let $k = 2\ell + 1$ for $\ell\geq 1$. All our $\ell_0$-samplers will have failure probability $n^{-3}$.
\begin{itemize}
    \item \textbf{Phase 1: Initializing clusters}

    We first sample a set of vertices $\mathcal C_0$ independently with probability $n^{-1/k}$. Then, on the first pass, we use $n$ copies of $\ell_0$-samplers, one for each vertex $u$, to sample an edge between $u$ and $\mathcal C_0$. If such an edge exists, we include $u$ in this cluster. For each vertex, we use an additional $O(n^{1+1/k}\log n)$ copies of $\ell_0$-samplers that sample random neighbors of the vertex.

    \item \textbf{Phase 2: Expanding clusters}

    For $\ell-1$ iterations from $i = 1,\dots,\ell-1$, we expand the clusters by one layer at a time as follows. Sample a set of clusters $\mathcal C_i\subseteq\mathcal C_{i-1}$ independently with probability $n^{-1/k}$ each. We then use $n$ copies of $\ell_0$-samplers to sample an edge between each vertex $u$ and a sampled cluster. If such an edge exists, we include $u$ in this cluster. For each vertex, we use an additional $O(n^{1+1/k}\log n)$ copies of $\ell_0$-samplers that sample random adjacent clusters of the vertex.

    \item \textbf{Phase 3: Connecting clusters}

    Finally, after the $\ell-1$ iterations, we include an edge between every pair of clusters in $\mathcal C_{\ell-1}$, using an $\ell_0$-sampler for each pair.
\end{itemize}

\paragraph{Algorithm for $k$ even.}
Let $k = 2\ell$ for $\ell\geq 2$.
\begin{itemize}
    \item \textbf{Phase 1: Initializing clusters}

    We first sample a set of vertices $\mathcal C_0$ independently with probability $n^{-1/k}$. Then, on the first pass, we use $n$ copies of $\ell_0$-samplers, one for each vertex $u$, to sample an edge between $u$ and $\mathcal C_0$. If such an edge exists, we include $u$ in this cluster. For each vertex, we use an additional $O(n^{1+1/k}\log n)$ copies of $\ell_0$-samplers that sample random neighbors of the vertex.

    \item \textbf{Phase 2: Expanding clusters}

    For $\ell-1$ iterations from $i = 1,\dots,\ell-1$, we expand the clusters by one layer at a time as follows. Sample a set of clusters $\mathcal C_i\subseteq\mathcal C_{i-1}$ independently with probability $n^{-1/k}$ each. We then use $n$ copies of $\ell_0$-samplers to sample an edge between each vertex $u$ and a sampled cluster. If such an edge exists, we include $u$ in this cluster. For each vertex, we use an additional $O(n^{1+1/k}\log n)$ copies of $\ell_0$-samplers that sample random adjacent clusters of the vertex.

    \item \textbf{Phase 3: Connecting clusters}

    Finally, after the $\ell-1$ iterations, we include an edge between every pair of clusters in $\mathcal C_{\ell-1}$, and $\mathcal C_{\ell-2}$ using an $\ell_0$-sampler for each pair.
\end{itemize}

We refer to \cite{baswana2007simple} for a proof of correctness. The space usage is only amplified by the use of $\ell_0$-samplers, which increases our bound by a $\log^3 n$ factor. The total number of passes in both cases is $\ell+1 = \floor{k/2}+1$, as claimed.
\end{proof}
\

\end{document}